\documentclass{revtex4}

\usepackage{amsmath}
\usepackage{amsthm}
\usepackage{amsfonts}

\newtheorem{lemma}{Lemma}
\newtheorem{proposition}{Proposition}

\theoremstyle{remark}

\theoremstyle{definition}

\begin{document}

\title{The Radial Wave Operator in Similarity Coordinates}

\author{Roland Donninger}
\email{roland.donninger@univie.ac.at}
\affiliation{Faculty of Physics, Gravitational Physics,  
University of Vienna,  
Boltzmanngasse 5, A-1090 Wien}

\begin{abstract}We present a rigorous functional analytic setting to study the radial
wave equation in similarity coordinates. 
As an application we analyse linear
stability of the fundamental self--similar solution of 
the wave equation with a focusing power
nonlinearity.
This is part one in a series of papers, see Ref. \onlinecite{ich1}, Ref. \onlinecite{ich2}. 
\end{abstract}

\maketitle

\section{Introduction}

\subsection{Motivation}
One of the most interesting properties of nonlinear wave equations is
the development of singularities in finite time for solutions starting from
smooth initial data.
In a typical situation one observes that small initial data lead to dispersion
whereas for large data the solution ceases to exist after a finite time by
forming a singularity.
An interesting mathematical question in this respect concerns the nature of the
breakdown for generic large data solutions.
It turns out that singularity formation is often connected to the 
existence of special
solutions.
In particular, for certain systems it has been observed that the breakdown
occurs via a self--similar solution: Generic large data evolutions approach a
certain self--similar solution and blow up.
Such a knowledge is often based on numerical studies.
We remark that the convergence to the special solution is 
due to dispersion
and can in general only hold in a local sense since the total energy is
conserved.
Obviously, the stability of special solutions plays a crucial role for
understanding the dynamics. 

As a first step in the analysis of such a system one studies linear 
stability of the solution which is expected to act as an attractor.
However, unlike in the case of static solutions, the linearization around a
self--similar solution yields an equation which 
depends explicitly on
time.
An elegant way to overcome this difficulty is to introduce
adapted coordinates which transform the problem of self--similar blow up into an
asymptotic stability problem. 
Such similarity coordinates have been successfully used 
by various authors (e.g. Ref. \onlinecite{Bizon1}, Ref. \onlinecite{Merle1}, Ref. \onlinecite{Merle2}, 
Ref. \onlinecite{Merle3}, Ref. \onlinecite{Merle4}, Ref. \onlinecite{Galaktionov1})
to study
self--similar blow up behaviour for the semilinear wave equation. 

The main motivation of the present paper is to obtain a systematic 
functional analytic initial value formulation of the wave equation in similarity
coordinates.
For simplicity we restrict ourselves to the radial case.
The main difficulty one encounters is the fact that the spatial part of the
wave operator in similarity coordinates is not normal.
Therefore, standard spectral theory is not applicable.
However, we are still able to obtain a well--posed initial value formulation by
means of semigroup theory.
This approach leads to interesting non self--adjoint spectral problems.
In particular, we show that the generator of the semigroup possesses a continuum of eigenvalues
filling a complex left half plane.
As an application we consider the wave equation with a power
nonlinearity where we study the linearization around the fundamental
self--similar solution and show that it is linearly stable.
This is achieved by a detailed study of the spectrum of the associated linear
operator.
This operator has exactly one unstable eigenvalue which is shown to emerge from
the time translation symmetry of the orginal problem.
We get rid of this instability by defining an appropriate 
projection that
removes the unstable eigenvalue from the spectrum.
Then we show that the resulting operator generates a semigroup and by analysing
its resolvent we deduce
appropriate growth bounds to conclude linear stability of the fundamental
self--similar solution.
The framework developed in the present work is fundamental for the follow--up papers Ref. \onlinecite{ich1} and Ref. \onlinecite{ich2} where the full nonlinear stability of the fundamental self--similar solution is proved.

\subsection{Notations}
Throughout this paper we will make use of the following notations. 
For a Hilbert space
$H$ we denote by $\mathcal{B}(H)$ the Banach space of linear bounded operators
on $H$. For a closed linear operator $L: \mathcal{D}(L) \subset H \to H$ we set
$R_L(\lambda):=(\lambda-L)^{-1}$ if $(\lambda-L)$ is bijective.
The spectrum $\sigma(L)$ of $L$ is decomposed as $\sigma(L)=\sigma_p(L) \cup
\sigma_c(L) \cup \sigma_r(L)$ where 
\begin{itemize}
\item $\sigma_p(L)$ is the set of all eigenvalues of $L$,
\item $\lambda \in \sigma_c(L)$ if $\lambda-L$ is injective but not surjective
and the range of $\lambda-L$ is dense in $H$,
\item $\lambda \in \sigma_r(L)$ if $\lambda-L$ is injective and the range of
$\lambda-L$ is not dense in $H$.
\end{itemize}
The resolvent set of $L$ is denoted by $\rho(L)$.
In order to improve readability, we write vectors as bold face letters and
the individual components are numbered using lower indices, e.g. 
$\mathbf{u}=(u_1,u_2)^T$.
Finally, the letter $C$ stands for a positive constant, however, this constant
is not supposed to have the same value at each occurence.

\section{The Free Wave Equation}

\subsection{Similarity Coordinates}
\label{sec_sim}
Consider the free wave equation in spherical symmetry, i.e.
$$ \tilde{\psi}_{tt}-\tilde{\psi}_{rr}-\frac{2}{r}\tilde{\psi}_r=0 $$
together with initial data $\tilde{\psi}(0,\cdot), \tilde{\psi}_t(0,\cdot)$.
The simple transformation $\tilde{\psi} \mapsto \psi$ where
$\psi(t,r):=r\tilde{\psi}(t,r)$ yields the one--dimensional wave equation 
$\psi_{tt}-\psi_{rr}=0 $ 
together with the regularity condition $\psi(t,0)=0$ for all $t$.
Our aim is to study this equation in the coordinate system $(\tau, \rho)$, 
defined by $\tau:=-\log(T-t)$ and $\rho:=\frac{r}{T-t}$, which is adapted to 
self--similarity and we restrict ourselves to the backward lightcone of the
point $(T,0)$ for a $T>0$.
Note that the transformation $(t,r) \mapsto (\tau,\rho)$ maps the
cone $\{(t,r): 0\leq t \leq T, 0 \leq r \leq T-t\}$ to the cylinder 
$\{(\tau,\rho): \tau \geq -\log T, \rho \in
[0,1]\}$ and thus, $\tau \to \infty$ corresponds to $t \to T-$.
The inverse transformation is given by $t=T-e^{-\tau}$ and $r=\rho e^{-\tau}$.
Furthermore, we note that the derivatives are given by $\partial_t=e^\tau
(\partial_\tau +\rho \partial_\rho)$ and $\partial_r=e^\tau \partial_\rho$.
We introduce the new dependent variables 
$\psi_1:=\psi_t$ and $\psi_2:=\psi_r$.
Then, $\Psi:=(\psi_1,\psi_2)^T$ satisfies
$$ \partial_t \Psi=\left ( \begin{array}{cc}
0 & 1 \\
1 & 0 \end{array} \right ) \partial_r \Psi $$
with the regularity condition $\psi_1(t,0)=0$ for all $t$.
In similarity coordinates $(\tau,\rho)$ this system reads
\begin{equation}
\label{eq_evol}
\partial_{\tau} \Phi=\left ( \begin{array}{cc}
-\rho & 1 \\
1 & -\rho \end{array} \right ) \partial_{\rho} 
\Phi
\end{equation}
where $\Phi(\tau,\rho)=\Psi(T-e^{-\tau}, \rho e^{-\tau})$.
The field $\phi(\tau,\rho):=\psi(T-e^{-\tau},\rho e^{-\tau})$ can be 
reconstructed from $\phi_2$ by noting that
$\phi_\rho(\tau, \rho)=e^{-\tau} \phi_2(\tau,\rho)$ which
yields
\begin{equation}
\label{eq_phi}
\phi(\tau,\rho)=e^{-\tau} \int_{0}^\rho \phi_2(\tau,\xi)d\xi.
\end{equation}

Now suppose we are able to deduce a growth bound for $\Phi$ of the form
$\|\Phi(\tau,\cdot)\| \leq Ce^{\mu \tau}$ for a $\mu \in \mathbb{R}$ where 
$\|\cdot\|$ is the norm on
$L^2(0,1)^2$.
What does this imply for the original system?
Consider the energy of the original field $\tilde{\psi}$ in
the backward lightcone of the point $(T,0)$ which is given by
$$ \int_0^{T-t} \left [ \psi_t^2(t,r)+\left (\psi_r(t,r)
-\frac{\psi(t,r)}{r}
\right )^2 \right ]dr $$
in terms of $\psi$.
Switching to similarity coordinates $(\tau,\rho)$ this expression reads
$$ E_\phi(\tau):=e^{-\tau} \int_0^1 \left [ \phi_1^2(\tau,\rho)
+\left (\phi_2(\tau,\rho)- \frac{e^\tau}{\rho}\phi(\tau,\rho)
\right )^2 \right ]d\rho. $$
Applying the Hardy inequality to Eq. (\ref{eq_phi}) we conclude that 
$$ \int_0^1 \left | \frac{\phi(\tau, \rho)}{\rho} \right |^2 d\rho \leq C e^{-2 \tau}\|\phi_2(\tau,\cdot)\|^2_{L^2(0,1)} $$
and the growth bound $\|\Phi(\tau,\cdot)\|
\leq Ce^{\mu \tau}$ implies 
$E_\phi(\tau)\leq
Ce^{2(\mu-\frac{1}{2})\tau}$.
Thus, if $\mu \leq \frac{1}{2}$, $E_\phi(\tau)$ stays bounded 
for all $\tau$ and the zero solution is
stable in the energy space.
Of course, for the free wave equation this stability is trivial,
however, for our purposes it is useful to "reprove" it by deducing an
appropriate growth bound for $\Phi$ which we are going to do in the following.

Furthermore, from 
$\|\phi_2(\tau,\cdot)\|_{L^2(0,1)}\leq
Ce^{\mu \tau}$ and Eq. (\ref{eq_phi}) we conclude
$$ |\phi(\tau,\rho)| \leq e^{-\tau}\|\phi_2(\tau,\cdot)\|_{L^2(0,1)}\leq
Ce^{(\mu-1)\tau} $$
by Cauchy--Schwarz and thus, $\|\phi(\tau,\cdot)\|_{L^\infty(0,1)} \leq
Ce^{(\mu-1)\tau}$.
Therefore, the zero solution is \emph{asymptotically} stable in the $L^\infty$
norm if $\mu<1$.
In the original coordinates $(t,r)$ this estimate translates into
$\|\psi(t,\cdot)\|_{L^\infty(0,T-t)}\leq C(T-t)^{1-\mu}$.

\subsection{Operator Formulation and Generation of Semigroup}
\label{sec_opform}
In order to deduce growth bounds for $\Phi$ we apply semigroup theory, i.e. we
derive an operator formulation for Eq. (\ref{eq_evol}).
Let $\mathcal{H}:=L^2(0,1)^2$ and define $\mathcal{D}(\tilde{L}_0):=\{\mathbf{u} \in
C^1[0,1]^2: u_1(0)=0\}$, $\tilde{L}_0 \mathbf{u}:=(-\rho u_1'+u_2', u_1'-\rho u_2')^T$.
Then $\tilde{L}_0: \mathcal{D}(\tilde{L}_0)\subset \mathcal{H} \to \mathcal{H}$ is a densely
defined linear operator.
The inner product and the norm on $\mathcal{H}$ are denoted by $(\cdot | \cdot)$
and $\|\cdot\|$, respectively.
Thus, 
\begin{equation}
\label{eq_evolop}
\frac{d}{d\tau} \Phi(\tau)=\tilde{L}_0 \Phi(\tau) 
\end{equation}
for a function $\Phi: [\tau_0,\infty) \to \mathcal{H}$ is an operator formulation of 
Eq. (\ref{eq_evol}) where $\tau_0:=-\log T$, i.e. $\tau=\tau_0$ corresponds to
$t=0$.
In what follows we prove that the closure of $\tilde{L}_0$ generates a strongly
continuous one--parameter semigroup $S_0: [0,\infty) \to \mathcal{B}(\mathcal{H})$
and thus, the unique solution of Eq. (\ref{eq_evolop}) is given by
$\Phi(\tau+\tau_0)=S_0(\tau)\Phi(\tau_0)$.

\begin{lemma}
\label{lem_qaccr}
The operator $\tilde{L}_0$ satisfies $\mathrm{Re}(\tilde{L}_0 \mathbf{u}|\mathbf{u})\leq
\frac{1}{2}\|\mathbf{u}\|^2$ for all $\mathbf{u} \in \mathcal{D}(\tilde{L}_0)$.
\end{lemma}

\begin{proof}
Integration by parts yields
\begin{multline*}
\mathrm{Re}(\tilde{L}_0 \mathbf{u} | \mathbf{u})=-\sum_{j=1}^2 \mathrm{Re} 
\int_0^1 \rho u'_j(\rho) \overline{u_j(\rho)}
d\rho + \mathrm{Re}\int_0^1 (u'_2(\rho) \overline{u_1(\rho)}
+u_1'(\rho) \overline{u_2(\rho)})d\rho \\
=\frac{1}{2}\|\mathbf{u}\|^2 - \frac{1}{2}|u_1(1)|^2-\frac{1}{2} |u_2(1)|^2
+\mathrm{Re}\overline{u_1(1)}u_2(1)\\
+\mathrm{Re} \left (2i\mathrm{Im}\int_0^1 u_1'(\rho)  \overline{u_2(\rho)}  
d\rho \right ) 
\leq \frac{1}{2} \|\mathbf{u}\|^2.
\end{multline*}
\end{proof}

\begin{lemma}
\label{lem_range}
The range of $1-\tilde{L}_0$ is dense in $\mathcal{H}$.
\end{lemma}

\begin{proof}
Let $\mathbf{f} \in C^\infty[0,1]^2$ and define $\mathbf{u}$ by
$$ u_2(\rho):=\frac{1}{1-\rho^2}\int_\rho^1 F(\xi)d\xi $$
and $u_1(\rho):=\rho u_2(\rho)-\int_0^\rho f_2(\xi)d\xi$ where 
$F(\rho):=f_1(\rho)+\rho f_2(\rho)+\int_0^\rho f_2(\xi)d\xi$.
An application of de l'Hospital's rule shows that $u_2 \in C^1[0,1]$ 
and obviously, $u_1 \in C^1[0,1]$, $u_1(0)=0$.
Hence, $\mathbf{u} \in \mathcal{D}(\tilde{L}_0)$ and a direct 
computation shows $(1-\tilde{L}_0)\mathbf{u}=\mathbf{f}$.
Thus, the claim follows from the density of $C^\infty[0,1]^2$
in $\mathcal{H}$.
\end{proof}

\begin{proposition}
\label{prop_L0gen}
The operator $\tilde{L}_0$ is closable and its closure $L_0$ generates a
strongly continuous semigroup $S_0: [0,\infty) \to \mathcal{B}(\mathcal{H})$
satisfying $\|S_0(\tau)\|\leq e^{\frac{1}{2}\tau}$.
\end{proposition}

\begin{proof}
The claim follows from Lemmas \ref{lem_qaccr}, \ref{lem_range} and the
Lumer--Phillips Theorem (see e.g. Ref. \onlinecite{Beyer2007}, p. 56, Theorem 4.2.6).
\end{proof}

Hence, this semigroup description yields the appropriate functional analytic
setting to study the wave equation in similarity coordinates. 
According to the discussion in Sec. \ref{sec_sim}, the growth bound
$\|S_0(\tau)\|\leq e^{\frac{1}{2}\tau}$ implies the well--known stability of the
zero solution in energy space.
Furthermore, we conclude that the zero solution is asymptotically stable
with respect
to the $L^\infty$ norm.

\subsection{Spectral Analysis of the Generator}
\label{sec_freespec}
We study the spectrum of $L_0$. 
To this end we need a more explicit description of $\mathcal{D}(L_0)$.
We define the formal matrix differential expression $\ell_0 \mathbf{u}:=(-\rho
u_1'+u_2', u_1'-\rho u_2')^T$, i.e. $\ell_0$ generates $L_0$.

\begin{lemma}
The domain of $L_0$ is given by
$$ \mathcal{D}(L_0)=\{\mathbf{u} \in \mathcal{H}: \mathbf{u} \in
H^1_\mathrm{loc}(0,1)^2, \ell_0 \mathbf{u} \in \mathcal{H}, u_1(0)=0\}. $$
\end{lemma}

\begin{proof}
Suppose $\mathbf{u} \in \mathcal{D}(L_0)$. Then there exists a sequence
$(\mathbf{u}_j) \subset \mathcal{D}(\tilde{L}_0)$ such that $\mathbf{u}_j \to
\mathbf{u}$ and $\tilde{L}_0 \mathbf{u}_j \to L_0\mathbf{u}$ in $\mathcal{H}$.
This implies that $((1-\rho^2)u_{1j}')$ and $((1-\rho^2)u_{2j}')$ are Cauchy
sequences in $L^2(0,1)$.
Thus, $u_1, u_2 \in H^1(0,1-\delta)$ for any $\delta \in (0,1)$ and the Sobolev
embedding $H^1(0,1-\delta) \hookrightarrow C[0,1-\delta]$ yields $u_1(0)=0$.

Conversely, if $\mathbf{u} \in \mathcal{H} \cap H^1_\mathrm{loc}(0,1)^2$, $\ell_0
\mathbf{u} \in \mathcal{H}$ and $u_1(0)=0$ then 
$(1-\ell_0)\mathbf{u} \in \mathcal{H}$ and we define $\mathbf{v}:=R_{L_0}(1)(1-\ell_0)\mathbf{u}$.
Then $\mathbf{v} \in \mathcal{D}(L_0)$ and by definition we have
$(1-\ell_0)\mathbf{v}=(1-\ell_0)\mathbf{u}$. From the first part of the proof we conclude that
$v_1(0)=0$ and a short calculation shows that $(1-\ell_0)(\mathbf{u}-\mathbf{v})=0$ necessarily
implies $\mathbf{u}=\mathbf{v} \in \mathcal{D}(L_0)$. 
\end{proof} 

Next we show that the analysis of the spectrum of $L_0$ can be reduced 
to the study of the invertibility of an operator--valued function.
For $\lambda \in \mathbb{C}$ we define the formal differential expression
$t_0(\lambda)$ by
$$ t_0(\lambda)u(\rho):=-(1-\rho^2)u''(\rho)+2 \lambda \rho
u'(\rho)+\lambda(\lambda-1)u(\rho) $$
and set $\mathcal{D}(T_0(\lambda)):=\{u \in H^1(0,1): u \in H^2_\mathrm{loc}(0,1), t_0(\lambda)u\in L^2(0,1), u(0)=0\}$, $T_0(\lambda)u:=t_0(\lambda)u$.
The following result shows that $\sigma(L_0)$ is completely determined by $T_0(\lambda)$.

\begin{proposition}
\label{prop_spec}
The operator $\lambda-L_0$ for a $\lambda \in \mathbb{C}$ is bounded invertible if and only if $T_0(\lambda)$ is invertible. Furthermore, $\lambda \in \sigma_p(L_0)$ if and only if $\dim \ker T_0(\lambda)=1$.  In this case, $\mathbf{u}=(u_1,u_2)^T$ defined by $u_1(\rho):=\rho u'(\rho)+(\lambda-1)u(\rho)$ and $u_2:=u'$ for $u \in \ker T_0(\lambda)$, $u \not= 0$, is an eigenfunction of $L_0$.
\end{proposition}

\begin{proof}
Suppose $\lambda \in \sigma_p(L_0)$ and $\mathbf{u}$ is the associated eigenfunction.
Then, $u_1'=\lambda u_2+\rho u_2'$ and this implies $u_1(\rho)=\rho u_2(\rho)+(\lambda-1)\int_0^\rho u_2(\xi)d\xi$ since $u_1(0)=0$. 
Inserting in $\lambda u_1+\rho u_1'-u_2'=0$ yields
$$ -(1-\rho^2)u_2'(\rho)+2\lambda \rho u_2(\rho)+\lambda (\lambda-1) \int_0^\rho u_2(\xi)d\xi=0. $$
Set $u(\rho):=\int_0^\rho u_2(\xi)d\xi$. Then, $u \in H^1(0,1)\cap H^2_\mathrm{loc}(0,1)$, $u(0)=0$ and $t_0(\lambda)u=0$. Thus, $u \in \ker T_0(\lambda)$.

Conversely, let $u \in \ker T_0(\lambda)$, $u \not=0$, and define $u_1(\rho):=\rho u'(\rho)+(\lambda-1)u(\rho)$, $u_2:=u'$.
Then, $\mathbf{u}=(u_1,u_2)^T \in \mathcal{H} \cap H^1_\mathrm{loc}(0,1)^2$, 
$u_1(0)=0$ and $\ell_0 \mathbf{u}=\lambda \mathbf{u}$. 
This shows $\mathbf{u} \in \ker (\lambda-L_0)$ and thus, 
$\lambda \in \sigma_p(L_0)$.

Suppose $\lambda-L_0$ is surjective and set $\mathbf{f}:=(f,0)^T \in \mathcal{H}$. Then, there exists a $\mathbf{u} \in \mathcal{D}(L_0)$ such that $(\lambda-L_0)\mathbf{u}=\mathbf{f}$.
As before, this implies that $u(\rho):=\int_0^\rho u_2(\xi)d\xi$ belongs to $\mathcal{D}(T_0(\lambda))$ and satisfies $T_0(\lambda)u=f$. Thus, $T_0(\lambda)$ is surjective.

Conversely, if $T_0(\lambda)$ is surjective, we can find a 
$u \in \mathcal{D}(T_0(\lambda))$ satisfying 
$T_0(\lambda)u(\rho)=f_1(\rho)+\rho f_2(\rho)+\lambda \int_0^\rho f_2(\xi)d\xi$ 
for any $\mathbf{f}=(f_1,f_2)^T \in \mathcal{H}$.
Defining $\mathbf{u}$ by $u_1(\rho):=\rho
u'(\rho)+(\lambda-1)u(\rho)-\int_0^\rho f_2(\xi)d\xi$ and $u_2:=u'$ we observe
that $\mathbf{u} \in \mathcal{D}(L_0)$ and $(\lambda-L_0)\mathbf{u}=\mathbf{f}$
which shows surjectivity of $\lambda-L_0$.

Thus, we have shown that $\lambda-L_0$ is bijective if and only if $T_0(\lambda)$ is bijective. The closed graph theorem implies that $(\lambda-L_0)^{-1}$ is bounded if it exists. Finally, by basic ODE theory we conclude that $\dim \ker T_0(\lambda)$ is at most one--dimensional.
\end{proof}

With Proposition \ref{prop_spec} at hand we can easily calculate the spectrum of $L_0$ by solving a second order ODE.
The complete characterization of $\sigma(L_0)$ is given in the following Lemma.

\begin{lemma}
The spectrum of $L_0$ is given by $\sigma(L_0)=
\{\lambda \in \mathbb{C}: \mathrm{Re}\lambda \leq \frac{1}{2}\}$ where
$\sigma_p(L_0)=\{\lambda \in \mathbb{C}: \mathrm{Re}\lambda < \frac{1}{2}\}$, $\sigma_c(L_0)=\{\lambda \in \mathbb{C}: \mathrm{Re}\lambda=\frac{1}{2}\}$, $\sigma_r(L_0)=\emptyset$.
\end{lemma}

\begin{proof}
First of all we note that the estimate $\|S_0(\tau)\| \leq e^{\frac{1}{2}\tau}$
for the semigroup $S_0$ generated by $L_0$ implies
that $\lambda \in \rho(L_0)$ if $\mathrm{Re}\lambda>\frac{1}{2}$ (see e.g.
Ref. \onlinecite{Engel2000}, p. 55, Theorem 1.10 (ii)) and hence, we restrict ourselves to $\mathrm{Re}\lambda \leq \frac{1}{2}$.

The equation $t_0(\lambda)u=0$ can be solved explicitly and the solution 
$u_0(\cdot, \lambda)$, which satisfies $u_0(0,\lambda)=0$, is given by
$u_0(\rho,\lambda)=(1-\rho)^{1-\lambda}-(1+\rho)^{1-\lambda}$.
Thus, having Proposition \ref{prop_spec} in mind we can immediately read off the 
point spectrum of $L_0$. 
We observe that
$u_0(\cdot, \lambda) \in \mathcal{D}(T_0(\lambda))$ 
if and only if $\mathrm{Re}\lambda < \frac{1}{2}$ which shows that any 
$\lambda$ with real part smaller than $\frac{1}{2}$ is an eigenvalue of $L_0$. 
Since the spectrum is always closed, we conclude that $\sigma(L_0)=\{\lambda \in
\mathbb{C}: \mathrm{Re}\lambda \leq \frac{1}{2}\}$.
Furthermore, by Ref. \onlinecite{Engel2000}, p. 242, Proposition 1.10 we know that the 
topological boundary $\frac{1}{2}+i\mathbb{R}$ of $\sigma(L_0)$ is contained in 
the approximate point spectrum which is given by 
$\sigma_p(L_0) \cup \sigma_c(L_0)$ and we conclude that 
$\sigma_c(L_0)=\frac{1}{2}+i\mathbb{R}$.
\end{proof}

We have the interesting situation that $L_0$ possesses a continuum of eigenvalues.
We remark that there exists a special subset $\{0,-1,-2,\dots\}$
of eigenvalues with analytic 
eigenfunctions.
However, from the point of view of semigroup theory there is no reason to
consider those analytic eigenfunctions as distinguished.
As a consequence we see that for any 
$\varepsilon>0$ there exist solutions of eq.
(\ref{eq_evolop}) that grow like
$\exp \left ((\frac{1}{2}-\varepsilon)\tau \right)$ for $\tau \to \infty$ which 
shows that the growth bound $\|S_0(\tau)\|\leq e^{\frac{1}{2}\tau}$ is sharp.
This has implications for the original system in $(t,r)$ coordinates as well.
Suppose one prescribes initial data determined by a nonanalytic mode solution.
Although these modes are only defined for $r \leq T-t$,
one can extend them in a $H^1$--fashion to yield perfectly admissible
finite energy initial data.
By causality the time development inside the past lightcone of $(t,r)=(T,0)$ is
independent of how the initial data have been extended and
exactly determined by the behaviour of the mode solution.
Hence, these modes are clearly seen to participate in generic time evolutions.

\subsection{An Estimate for $T_0^{-1}(\lambda)$}
The operator $L_0$ generates the strongly continuous semigroup $S_0$ which
satisfies $\|S_0(\tau)\|\leq e^{\frac{1}{2}\tau}$ and this implies the resolvent estimate
$\|R_{L_0}(\lambda)\|\leq \frac{1}{\mathrm{Re}\lambda-\frac{1}{2}}$ for all
$\lambda \in \mathbb{C}$ with $\mathrm{Re}\lambda>\frac{1}{2}$ 
(cf. e.g. Ref. \onlinecite{Engel2000}, p. 55, Theorem 1.10 (iii)).
The resolvent $R_{L_0}(\lambda)$ can be given in terms of $T_0^{-1}(\lambda)$.
To this end we define the operator $B(\lambda): \mathcal{H} \to L^2(0,1)$ by
$B(\lambda)\mathbf{f}(\rho):=f_1(\rho)+\rho f_2(\rho)+\lambda \int_0^\rho
f_2(\xi)d\xi$.
A straightforward calculation shows that
$$ R_{L_0}(\lambda)\mathbf{f}(\rho)=\left ( \begin{array}{c} \rho 
(T_0^{-1}(\lambda)B(\lambda)\mathbf{f})'(\rho)+(\lambda
-1)T_0^{-1}(\lambda)B(\lambda)\mathbf{f}(\rho)-\int_0^\rho f_2(\xi)d\xi \\
(T_0^{-1}(\lambda)B(\lambda)\mathbf{f})'(\rho)
\end{array} \right ). $$
The resolvent estimate implies 
$\|(T_0^{-1}(\lambda)B(\lambda)\mathbf{f})'\|_{L^2(0,1)}\leq
\frac{1}{\mathrm{Re}\lambda-\frac{1}{2}}$ and the first component of 
$R_{L_0}(\lambda)\mathbf{f}$
then yields 
\begin{equation}
\label{eq_estT0}
\|T_0^{-1}(\lambda)B(\lambda)\mathbf{f}\|_{L^2(0,1)} \leq \frac{1}{|\lambda-1|}\left
(\frac{2}{\mathrm{Re}\lambda-\frac{1}{2}}+1 \right )\|\mathbf{f}\|
\end{equation} 
for all $\lambda$ with $\mathrm{Re}\lambda>\frac{1}{2}$, $\lambda \not= 1$.
This estimate will be extremely useful later on.

\section{The Wave Equation With a Power Nonlinearity}
\subsection{Definition of the System, Known Results}
Consider the semilinear focusing wave equation
\begin{equation}
\label{eq_semlinwave}
\chi_{tt}-\Delta \chi-\chi^p=0
\end{equation}
for $\chi: \mathbb{R}\times \mathbb{R}^3 \to \mathbb{R}$ where
$p>1$ is an odd
integer.
This equation has been subject to many studies on nonlinear wave equations 
due to its simplicity. 
For the local well--posedness we refer to Lindblad and Sogge Ref. \onlinecite{Lindblad1}
and references therein.

Neglecting the Laplacian immediately leads to the homogeneous--in--space solution
$$ \chi_0(t,x)=c_0^{1/(p-1)}(T-t)^{-2/(p-1)}, \:\: c_0=\frac{2(p+1)}{(p-1)^2} $$
which, together with smooth cut--off and finite speed of propagation, provides 
an example of a solution which starts from smooth compactly supported initial 
data and blows up for $t \to T-$. 
Based on numerics, it is conjectured that this solution describes the generic 
blow up scenario (cf. Ref. \onlinecite{Bizon1}).
For the system in one space dimension, this conjecture has been proved by 
Merle and Zaag Ref. \onlinecite{Merle3}.
Furthermore, the blow up behaviour
for certain ranges of $p$ and different space dimensions has been rigorously
analysed in Ref. \onlinecite{Merle1}, Ref. \onlinecite{Merle2},
Ref. \onlinecite{Merle4}.
In all of these studies the coordinate system $(\tau,\rho)$ plays an important role.
We also mention the work of Galaktionov and Pohozaev Ref. \onlinecite{Galaktionov1} where a
perturbative analysis in similarity coordinates is presented.

To analyse the stability of $\chi_0$
we linearize the problem by 
inserting the ansatz $\chi=\chi_0+\tilde{\psi}$ into Eq. (\ref{eq_semlinwave})
and neglecting higher order terms which leads to
\begin{equation}
\tilde{\psi}_{tt}(t,x)-\Delta \tilde{\psi}(t,x) -p c_0 (T-t)^{-2} \tilde{\psi}(t,x)=0.
\end{equation}
We restrict ourselves to spherically symmetric perturbations and write $r=|x|$. 
As in Sec. \ref{sec_sim}, the substitution $\tilde{\psi} \mapsto \psi$ 
where $\psi(t,r)=r\tilde{\psi}(t,r)$ 
yields
\begin{equation}
\label{eq_psilin}
\psi_{tt}(t,r)-\psi_{rr}(t,r)-pc_0 (T-t)^{-2} \psi(t,r)=0
\end{equation}
together with the regularity condition $\psi(t,0)=0$ for all $t$. 
As before, we intend to study this equation in similarity coordinates 
$(\tau, \rho)$
and introduce the new dependent variables 
$\psi_1:=\psi_t$ and $\psi_2:=\psi_r$.
Then, $\Psi:=(\psi_1,\psi_2)^T$ satisfies
$$ \partial_t \Psi(t,r)=\left ( \begin{array}{cc}
0 & 1 \\
1 & 0 \end{array} \right ) \partial_r \Psi(t,r)+pc_0(T-t)^{-2}\left ( \begin{array}{cc} 
0 & 1 \\
0 & 0 \end{array} \right ) \int_0^r \Psi(t,s)ds $$
with the regularity condition $\psi_1(t,0)=0$ for all $t$.
In similarity coordinates $(\tau,\rho)$ this system reads
\begin{equation}
\label{eq_linevol}
\partial_{\tau} \Phi(\tau, \rho)=\left ( \begin{array}{cc}
-\rho & 1 \\
1 & -\rho \end{array} \right ) \partial_{\rho} 
\Phi(\tau, \rho)+pc_0 \left ( \begin{array}{cc}
0 & 1 \\
0 & 0 \end{array} \right ) \int_0^\rho \Phi(\tau, \xi)d\xi 
\end{equation}
where $\Phi(\tau,\rho):=\Psi(T-e^{-\tau},\rho e^{-\tau})$.
Like in Sec. \ref{sec_sim}, the field $\phi(\tau,\rho):=
\psi(T-e^{-\tau},\rho e^{-\tau})$ is determined by Eq. (\ref{eq_phi}).

\subsection{Operator Formulation and Generation of Semigroup}

Let $\mathcal{H}:=L^2(0,1)^2$ and define $L' \in \mathcal{B}(\mathcal{H})$ by
$$ L' \mathbf{u}:=\left ( \begin{array}{c} pc_0 \int_0^\rho u_2(\xi)d\xi 
\\ 0 \end{array} \right ). $$
It follows that 
$$ \frac{d}{d\tau} \Phi(\tau)=(L_0+L')\Phi(\tau) $$
for a function $\Phi: [\tau_0, \infty) \to \mathcal{H}$ is an operator
formulation of Eq. (\ref{eq_linevol}) where $L_0$ is defined in Sec.
\ref{sec_opform} and $\tau_0=-\log T$.
Note also that $L'$ is compact which will be useful later on.

\begin{proposition}
\label{prop_Lgen}
The operator $L:=L_0+L'$ generates a strongly
continuous one--parameter semigroup $S: [0,\infty) \to \mathcal{B}(\mathcal{H})$
satisfying
$\|S(\tau)\| \leq e^{(\frac{1}{2}+pc_0)\tau}$.
\end{proposition}

\begin{proof}
The claim follows immediately from Proposition \ref{prop_L0gen}, the Bounded
Perturbation Theorem (see e.g. Ref. \onlinecite{Engel2000}, p. 158, Theorem 1.3) and
$\|L'\|\leq pc_0$.
\end{proof} 
  
\subsection{Spectral Analysis of the Generator}
In order to improve the weak growth estimate from Proposition \ref{prop_Lgen} we study
the spectrum of $L$.
Similar to the case of the free wave equation, the analysis of
$\sigma(L)$ can be reduced to the study of an operator pencil defined
on scalar functions.
As a first step we calculate the point spectrum of $L$.

\begin{lemma}
$\lambda \in \sigma_p(L)$ if and only if $\dim \ker
(T_0(\lambda)-pc_0)=1$.
\end{lemma}

The proof of this Lemma consists of an obvious modification of the first part 
of the proof of
Proposition \ref{prop_spec} and will therefore be omitted.

The equation $(t_0(\lambda)-pc_0)u=0$ has a regular singular point at $\rho=1$.
Solutions of this equation can be given in terms of the hypergeometric
function ${}_2F_1$.
Indeed, applying the substitution $\rho \mapsto z:=\rho^2$, 
$(t_0(\lambda)-pc_0)u=0$
transforms into the hypergeometric differential equation
\begin{equation}
\label{eq_hypgeom}
z(1-z)v''(z)+[c-(a+b+1)z]v'(z)-abv(z)=0 
\end{equation}
where $v(z):=u(\sqrt{z})$ and $a:=\frac{1}{4}(-1+2\lambda-\sqrt{1+4pc_0})$,
$b:=\frac{1}{4}(-1+2\lambda+\sqrt{1+4pc_0})$, $c:=\frac{1}{2}$.
Assume for the moment that $\lambda \not= 1$.
Two linearly independent solutions $v_1, \tilde{v}_1$ are given by 
$v_1(z)={}_2F_1(a, b; a+b+1-c;
1-z)$ and $\tilde{v}_1(z)=(1-z)^{c-a-b}{}_2F_1(c-a,c-b;c+1-a-b;1-z)$ (cf.
Ref. \onlinecite{Erdelyi1953}).
Hence, around $\rho=1$ there exist two linearly independent solutions $u_1$ and
$\tilde{u}_1$ of $(t_0(\lambda)-pc_0)u=0$ where $u_1$ is analytic and $\tilde{u}_1$ has
the form $\tilde{u}_1(\rho)=(1-\rho)^{1-\lambda}h(\rho)$ for an analytic function
 $h$ 
with $h(1)\not=0$. If $\lambda=1$ then $c-a-b=0$ and we have a
degenerate case. 
One solution is still given by $v_1$
as above and the second one diverges logarithmically 
for $z \to 1$. 

\begin{lemma} 
\label{lem_sigmapL}
The point spectrum of $L$ is given by 
$\sigma_p(L)=\sigma_p(L_0) \cup \{1+\frac{2}{p-1}\}$.
Moreover, we have $\sigma(L) \supset \sigma(L_0)$.
\end{lemma}

\begin{proof}
The asymptotic behaviour of solutions of $(t_0(\lambda)-pc_0)u=0$ for 
$\rho \to 1$ is the
same as for the analogous problem for the free wave equation 
(Sec. \ref{sec_freespec})
and thus, for $\mathrm{Re}\lambda<\frac{1}{2}$ both linearly independent
solutions are $H^1$ near $\rho=1$. This shows $\{\lambda \in \mathbb{C}:
\mathrm{Re}\lambda<\frac{1}{2}\} \subset \sigma_p(L)$.
Thus, we restrict ourselves to $\mathrm{Re}\lambda > \frac{1}{2}$.
Consider the following
solutions $v_1(z):={}_2F_1(a,b;a+b+1-c;1-z)$,
$v_0(z):=z^{1-c}{}_2F_1(a+1-c,b+1-c; 2-c; z)$ and
$\tilde{v}_0(z):={}_2F_1(a,b;c;z)$ of Eq. (\ref{eq_hypgeom}). 
Since $v_0$ and $\tilde{v}_0$ are linearly independent, there exist constants
$c_1$, $c_2$ such that $v_1=c_1 \tilde{v}_0+c_2 v_0$.
The solution $v_1$, which is analytic around $\rho=1$, satisfies the boundary 
condition $v_1(0)=0$ if and only if $c_1=0$ and exactly in this case, $\ker
(T_0(\lambda)-pc_0)\not=\{0\}$ since the other linearly independent solution
around $\rho=1$ does not belong to $H^1$.
The coefficient $c_1$ can be given in terms of the $\Gamma$--function (cf.
Ref. \onlinecite{Erdelyi1953}) and reads
$$ c_1=\frac{\Gamma(a+b+1-c)\Gamma(1-c)}{\Gamma(a+1-c)\Gamma(b+1-c)}. $$
Since the Gamma function does not have zeros, $c_1=0$ if and only if $a+1-c$ or
$b+1-c$ is a pole which yields the condition $a+1-c=-k$ or $b+1-c=-k$ for
$k=0,1,\dots$.
There is only one solution with $\mathrm{Re}\lambda>\frac{1}{2}$ and it
is given by $\lambda=1+\frac{2}{p-1}$.

Since the spectrum is closed, the line $\frac{1}{2}+i\mathbb{R}$ belongs to
$\sigma(L)$ and we conclude $\sigma(L) \supset \sigma(L_0)$.
\end{proof}

As a consequence of the compactness of $L'$ we are able to determine the
\emph{whole} spectrum of $L$ by considering solutions of the ODE
$(T_0(\lambda)-pc_0)u=0$.
This will be shown in the next Lemma.

\begin{lemma}
\label{lem_specL}
We have $\lambda \in \rho(L)$ if and only if $\lambda \in \rho(L_0)$ and 
$I-L'R_{L_0}(\lambda)$ is bounded
invertible.
Furthermore, $\sigma(L)=\sigma(L_0) \cup \{1+\frac{2}{p-1}\}$.
\end{lemma}

\begin{proof}
Let $\lambda \in \rho(L)$. Then, by Lemma \ref{lem_sigmapL}, 
$\lambda \in \rho(L_0)$ and the identity
\begin{equation}
\label{eq_ident}
\lambda-L=(I-L'R_{L_0}(\lambda))(\lambda-L_0)
\end{equation} shows that
$(I-L'R_{L_0}(\lambda))^{-1}$ exists and belongs to $\mathcal{B}(\mathcal{H})$.
Conversely, suppose $\lambda \in \rho(L_0)$ and 
$(I-L'R_{L_0}(\lambda))^{-1} \in \mathcal{B}(\mathcal{H})$.
Then, by using Eq. (\ref{eq_ident}) again we
conclude that $\lambda \in \rho(L)$.

Now let $\lambda \in \sigma(L)\backslash \sigma(L_0)$.
By the first part this implies $1 \in \sigma(L'R_{L_0}(\lambda))$.
From the compactness of $L'R_{L_0}(\lambda)$ we conclude that $1 \in
\sigma_p(L'R_{L_0}(\lambda))$
which means that there exists a $\mathbf{f} \in \mathcal{H}$, $\mathbf{f}
\not=0$, with
$(I - L'R_{L_0}(\lambda))\mathbf{f}=0$.
Set $\mathbf{u}:=(\lambda-L_0)^{-1}\mathbf{f}$. Then, $\mathbf{u} \not=0$,
$\mathbf{u} \in \mathcal{D}(L_0)=\mathcal{D}(L)$ and Eq. (\ref{eq_ident}) shows
that $(\lambda-L)\mathbf{u}=0$ which implies $\lambda \in \sigma_p(L)$.
Hence, any spectral value of $L$ which does not belong to $\sigma(L_0)$ is an
eigenvalue and by Lemma \ref{lem_sigmapL}, $1+\frac{2}{p-1}$ is the only
eigenvalue of $L$ outside $\sigma(L_0)$.
\end{proof}

We give an explanation how the existence of the single isolated (unstable)
eigenvalue
$\lambda_0:=1+\frac{2}{p-1}$ can be understood.
Consider a nonlinear equation $F(u)=0$ where $F: U \subset X \to Y$ is a
Fr\'echet differentiable nonlinear mapping from an open subset $U$ of a 
Banach space $X$ to a Banach space $Y$.
Suppose there exists a one--parameter family $\{u_s \in U: s \in (a,b)\}$
of solutions (i.e.
$F(u_s)=0$ for all $s \in (a,b)$) such that
the mapping $s \mapsto u_s: (a,b) \to X$ is (strongly) differentiable at $s_0
\in (a,b)$.
An application of the chain rule shows that
$0=\left . \frac{d}{ds}\right |_{s=s_0}F(u_s)=DF(u_{s_0})\left. \frac{d}{ds}\right|_{s=s_0}u_s$ and
therefore, $\left . \frac{d}{ds}\right |_{s=s_0}u_s \in X$ is a solution of the linearized
problem $DF(u_{s_0})u=0$ where $DF$ is the Fr\'echet derivative of $F$.
 
Note that we can choose the blow up time $T$ freely, hence, the solution
$\chi_0(t,r)=c_0^{1/(p-1)}(T-t)^{-2/(p-1)}$ is a one--parameter family of 
solutions rather than a single one.
A direct calculation shows that $\psi(t,r):=r \frac{d}{dT}\chi_0(t,r)$ solves
the linearized problem eq.
(\ref{eq_psilin}) and in similarity coordinates, $\psi_r$ is given by
$Ce^{\lambda_0 \tau}$.
Therefore, the existence of the unstable eigenvalue $\lambda_0$ 
is a direct consequence of the time
translation symmetry of the original problem.
Actually, we are only interested in "stability modulo this symmetry" and
in the next section we will show how to make this idea rigorous.
 
\subsection{Spectral Projection on the Stable Subspace}

We want to restrict the set of admissible perturbations in such a way that the
instability caused by the time translation symmetry does not contribute to the
linear time evolution.
The spectrum of $L$ is separated in a bounded part (which consists of the single
eigenvalue $1+\frac{2}{p-1}$) and an unbounded part $\sigma(L_0)$.
Let $\Gamma$ be a circle in the complex plane with center $1+\frac{2}{p-1}$ and
radius smaller than 1, i.e. $\Gamma \subset \rho(L)$.
Then we can define the spectral projection $P \in \mathcal{B}(\mathcal{H})$ by
$$ P:=\frac{1}{2 \pi i}\int_\Gamma R_L(\lambda)d\lambda $$
which projects onto the closed subspace 
$\mathcal{M}:=P\mathcal{H} \subset \mathcal{H}$.
The operator $P$ commutes with $L$ and we have the decomposition
$\mathcal{H}=\mathcal{M} \oplus \mathcal{N}$ where
$\mathcal{N}:=(I-P)\mathcal{H}$. 
Furthermore, $L$ is decomposed into two parts $L_\mathcal{M}:=L|_{\mathcal{D}(L)
\cap \mathcal{M}}$
and
$L_\mathcal{N}:=L|_{\mathcal{D}(L) \cap \mathcal{N}}$ acting on $\mathcal{M}$ and $\mathcal{N}$, respectively,
and $\sigma(L_\mathcal{M})=\{1+\frac{2}{p-1}\}$ whereas
$\sigma(L_\mathcal{N})=\sigma(L_0)$ 
(see Ref. \onlinecite{Kato1980}, p. 178, Theorem 6.17). 
However, we remark that the eigenvalues with analytic eigenfunctions are not the
same for $L_0$ and $L_\mathcal{N}$.
$L_\mathcal{N}$ inherits all the nice properties of $L$, i.e. $L_\mathcal{N}$ is
densely defined, closed and the resolvent of $L_\mathcal{N}$ is given by
$R_L(\lambda)|_\mathcal{N}$.
Thus, $L_\mathcal{N}$ generates a strongly continuous semigroup $S_\mathcal{N}$
on $\mathcal{N}$ satsfying $\|S_\mathcal{N}(\tau)\|\leq
e^{(\frac{1}{2}+pc_0)\tau}$.
The semigroup $S_\mathcal{N}$ describes the linear time evolution 
"modulo
the instability caused by the time translation symmetry of the original 
problem".
However, since the generator $L_\mathcal{N}$ is not normal, we cannot directly
obtain a sharp growth bound for $S_\mathcal{N}$ from the spectral properties
of $L_\mathcal{N}$.

The interesting notion concerning the qualitative behaviour of $S_\mathcal{N}$
is the growth bound $\omega_0(S_\mathcal{N})$ which is 
defined
by 
$$\omega_0(S_\mathcal{N}):=\inf\{s \in \mathbb{R}: \exists
C_s\geq 1: \|S_\mathcal{N}(\tau)\|\leq C_s e^{s \tau} \mbox{ for all }\tau > 0\}.$$
It is well--known (cf. e.g. Ref. \onlinecite{Engel2000}) that 
$$ \omega_0(S_\mathcal{N})=\inf \left \{\kappa>s(L_\mathcal{N}):
\sup_{\omega \in \mathbb{R}} \|R_L(\kappa+i\omega)|_\mathcal{N}\|<\infty \right
\} $$
where $s(L_\mathcal{N}):=\sup\{\mathrm{Re}\lambda: \lambda \in
\sigma(L_\mathcal{N})\}$ is the spectral bound of $L_\mathcal{N}$.
In our case we have $s(L_\mathcal{N})=\frac{1}{2}$ and
thus, we have to study the behaviour of $R_L(\lambda)$ for
$\mathrm{Re}\lambda>\frac{1}{2}$ and $|\mathrm{Im}\lambda| \to \infty$.

\begin{lemma}
\label{lem_resuni}
There exists a $C>0$ such that $\|R_L(\lambda)\|\leq C$ for all $\lambda \in
\rho(L)$.
\end{lemma}

\begin{proof}
Let $\mathrm{Re}\lambda > \frac{1}{2}$ and $\lambda \not= 1+\frac{2}{p-1}$. 
Applying Lemma \ref{lem_specL} and the identity
Eq. (\ref{eq_ident}) we obtain
$R_L(\lambda)=R_{L_0}(\lambda)(I-L'R_{L_0}(\lambda))^{-1}$.
The operator $L'R_{L_0}(\lambda)$ is given by
$$ L'R_{L_0}(\lambda)\mathbf{f}=\left ( \begin{array}{c}pc_0
T_0^{-1}(\lambda)B(\lambda)\mathbf{f} \\ 0 \end{array} \right ). $$
Thus, the estimate Eq. (\ref{eq_estT0}) implies 
$$ \|L'R_{L_0}(\lambda)\| \leq \frac{pc_0}{|\lambda-1|}\left (
\frac{2}{\mathrm{Re}\lambda-\frac{1}{2}}+1 \right ) $$
and this shows that $\|L'R_{L_0}(\lambda)\| \to 0$ for $|\mathrm{Im}\lambda|
\to \infty$.
We conclude that $I-L'R_{L_0}(\lambda) \to I$ and hence,
$(I-L'R_{L_0}(\lambda))^{-1} \to I$ for $|\mathrm{Im}\lambda| \to \infty$
uniformly.
Therefore, the claim follows from the uniform boundedness of
$\|R_{L_0}(\lambda)\|$ for $\mathrm{Re}\lambda>\frac{1}{2}$.
\end{proof}

Lemma \ref{lem_resuni} shows that $\omega_0(S_\mathcal{N})=\frac{1}{2}$ and
thus, the growth bound of the semigroup $S_\mathcal{N}$ coincides with the
spectral bound of its generator.
We immediately conclude the linear asymptotic stability of the solution $\chi_0$
with respect to the $L^\infty$ norm (cf. Sec. \ref{sec_sim}).
 
\section{Acknowledgments}
The author wants to thank Peter C. Aichelburg and Piotr Bizo\'n for helpful
discussions.
This work has been supported by the Austrian Fonds zur F\"orderung der
wissenschaftlichen Forschung (FWF) Project No. P19126.

\end{document}